\documentclass[12pt]{article} 

\def\UseBibLatex{1}

\makeatletter
\def\input@path{{styles/}}
\makeatother

\providecommand{\BibLatexMode}[1]{}
\providecommand{\BibTexMode}[1]{}

\renewcommand{\BibLatexMode}[1]{#1}
\renewcommand{\BibTexMode}[1]{}

\ifx\UseBibLatex\undefined%
  \renewcommand{\BibLatexMode}[1]{}
  \renewcommand{\BibTexMode}[1]{#1}
\fi

\BibLatexMode{%
   \usepackage[bibencoding=utf8,style=alphabetic,backend=biber]{biblatex}%
   \usepackage{sariel_biblatex}%
}

\usepackage[cm]{fullpage}%
\usepackage{amsmath}%
\usepackage{amssymb}%
\usepackage[table]{xcolor}%

\usepackage[amsmath,thmmarks]{ntheorem}%

\usepackage{titlesec}%
\usepackage{xcolor}%
\usepackage{mleftright}%
\usepackage{xspace}%
\usepackage{graphicx}
\usepackage{hyperref}%
\usepackage[inline]{enumitem}
\usepackage{hyperref}%
\usepackage[ocgcolorlinks]{ocgx2}

\hypersetup{%
      unicode,
      breaklinks,%
      colorlinks=true,%
      urlcolor=[rgb]{0.25,0.0,0.0},%
      linkcolor=[rgb]{0.5,0.0,0.0},%
      citecolor=[rgb]{0,0.2,0.445},%
      filecolor=[rgb]{0,0,0.4},
      anchorcolor=[rgb]={0.0,0.1,0.2}%
}

\titlelabel{\thetitle. }%

\theoremseparator{.}%

\theoremstyle{plain}%
\newtheorem{theorem}{Theorem}[section]

\newtheorem{lemma}[theorem]{Lemma}

\theoremstyle{plain}%
\theoremheaderfont{\sf} \theorembodyfont{\upshape}%
\newtheorem*{remark:unnumbered}[theorem]{Remark}%

\theoremheaderfont{\em}%
\theorembodyfont{\upshape}%
\theoremstyle{nonumberplain}%
\theoremseparator{}%
\theoremsymbol{\myqedsymbol}%
\newtheorem{proof}{Proof:}%

\providecommand{\emphind}[1]{}%
\renewcommand{\emphind}[1]{\emph{#1}\index{#1}}

\definecolor{blue25emph}{rgb}{0, 0, 11}

\providecommand{\emphic}[2]{}
\renewcommand{\emphic}[2]{\textcolor{blue25emph}{%
      \textbf{\emph{#1}}}\index{#2}}

\providecommand{\emphi}[1]{}%
\renewcommand{\emphi}[1]{\emphic{#1}{#1}}

\definecolor{almostblack}{rgb}{0, 0, 0.3}

\providecommand{\emphw}[1]{}%
\renewcommand{\emphw}[1]{{\textcolor{almostblack}{\emph{#1}}}}%

\providecommand{\emphOnly}[1]{}%
\renewcommand{\emphOnly}[1]{\emph{\textcolor{blue25emph}{\textbf{#1}}}}

\newcommand{\myqedsymbol}{\rule{2mm}{2mm}}
\newcommand{\SarielThanks}[1]{%
   \thanks{%
      School of Computing and Data Science; %
      University of Illinois; %
      201 N. Goodwin Avenue; %
      Urbana, IL, 61801, USA; %
      \href{mailto:spam@illinois.edu}{sariel@illinois.edu}; %
      \url{http://sarielhp.org/}.%
   #1%
   }%
}

\newcommand{\HLink}[2]{\hyperref[#2]{#1~\ref*{#2}}}
\newcommand{\HLinkSuffix}[3]{\hyperref[#2]{#1\ref*{#2}{#3}}}

\newcommand{\lemlab}[1]{\label{lemma:#1}}
\newcommand{\lemref}[1]{\HLink{Lemma}{lemma:#1}}%

\newcommand{\defrefY}[2]{\hyperref[def:#1]{#2}}

\providecommand{\eqlab}[1]{}%
\renewcommand{\eqlab}[1]{\label{equation:#1}}

\providecommand{\remove}[1]{}%
\newcommand{\Set}[2]{\left\{ #1 \;\middle\vert\; #2 \right\}}

\newcommand{\pth}[1]{\mleft(#1\mright)}%
\newcommand{\bpth}[1]{\mleft[#1\mright]}%

\newcommand{\ExC}{{\mathbb{E}}}

\renewcommand{\Re}{\mathbb{R}}%

\newlist{compactenumA}{enumerate}{5}%
\setlist[compactenumA]{itemsep=-0.5ex,topsep=0.5ex,partopsep=1ex,parsep=1ex,%
   label=(\Alph*)}%

\newlist{compactenuma}{enumerate}{5}%
\setlist[compactenuma]{itemsep=-0.5ex,topsep=0.5ex,partopsep=1ex,parsep=1ex,%
   label=(\alph*)}%

\newlist{compactenumI}{enumerate}{5}%
\setlist[compactenumI]{itemsep=-0.5ex,topsep=0.5ex,partopsep=1ex,parsep=1ex,%
   label=(\Roman*)}%

\newlist{compactenumi}{enumerate}{5}%
\setlist[compactenumi]{itemsep=-0.5ex,topsep=0.5ex,partopsep=1ex,parsep=1ex,%
   label=(\roman*)}%

\newlist{compactitem}{itemize}{5}%
\setlist[compactitem]{itemsep=-0.5ex,topsep=0.5ex,partopsep=1ex,parsep=1ex,%
   label=\ensuremath{\bullet}}%

\newcommand{\etal}{\textit{et~al.}\xspace}

\numberwithin{figure}{section}%
\numberwithin{table}{section}%
\numberwithin{equation}{section}%

\usepackage{mathcalb}

\newcommand{\ballC}{\mathcalb{b}}%
\newcommand{\ballY}[2]{\ballC\pth{#1,#2}}%
\newcommand{\sphereC}{\mathcalb{s}}%
\newcommand{\sphereY}[2]{\sphereC\pth{#1,#2}}%
\newcommand{\lenX}[1]{\left\| #1 \right\|}%

\BibLatexMode{%
   \bibliography{c_packed_sep} }

\begin{document}

\title{Separator for $c$-Packed Segments and Curves}

\author{Sariel Har-Peled\SarielThanks{}}

\date{\today}

\maketitle

\begin{abstract}
    We provide a simple algorithm for computing a balanced separator for a set of segments that is $c$-packed, showing that the separator cuts only $O(c)$ segments.  While the result was known before, arguably our new proof is simpler.
\end{abstract}

\section{Separator for $c$-packed segments}

As a reminder \cite{dhw-afdrc-12}, a set $S$ of segments in $\Re^d$ is \emphi{$c$-packed}, if for all $p \in\Re^d$ and $r > 0$, we have $\lenX{ S \sqcap \ballY{p}{r}} \leq c r$, where $\ballY{p}{r}$ is the ball of radius $r$ centered at $p$, $\lenX{ S \sqcap \ballY{p}{r}} = \sum_{ s \in S} \lenX{ s \cap \ballY{p}{r}}$, and $\lenX{s}$ is the length of $s$. Let $\#(p,r,S) = | \Set{ \sphereY{p}{r} \cap s}{s \in S} |$ denote the number of intersection points of the segments of $S$ with the sphere $\sphereY{p}{r}$ that bounds $\ballY{p}{r}$. Observe that
\begin{equation*}
    \lenX{S \sqcap \ballY{p}{r} }
    \geq
    \textstyle
    \int_{x=0}^r \#(p,x,S) \, dr .
\end{equation*}
As such, we get the following.
\begin{lemma}
    \lemlab{silly}%
    Let $S$ be a set of $c$-packed segments in $\Re^d$, then for any $p$, and $0 \leq a < b$, we have $\ExC_{x \in [a,b]}\bpth{\#(p,x,S)} \leq \frac{cb}{b-a}$.
\end{lemma}

The following separator result is from Deryckere \etal \cite{dgrsw-wsstc-25}, but our
proof is arguably simpler.

\begin{lemma}
    For a set $S$ of $n$ segments in $\Re^d$, one can compute, in expected linear time, a sphere $\sphereC= \sphereY{p}{x}$, such that the number of segments of $S$ intersecting $\sphereC$ is at most $O(c)$, and the number of of segments inside (resp. outside) $\sphereC$ is at least $\tfrac{n}{2c}$, where $c$ is a constant that depends only on the dimension $d$.
\end{lemma}
\begin{proof}
    This is by now a standard argument \cite{h-speps-13}. Let $V = V(S)$ be the set of all endpoints of the segments of $S$. Let $c= c'+1$, where $c' = 2^{O(d)}$ be the doubling constant of $\Re^d$ -- that is, the minimal number of balls of radius $1/2$ needed to cover a ball of radius $1$.  Let $\ballC = \ballY{p}{r}$ be the smallest ball containing $2n/c$ of the endpoints of the segments of $S$. For any sphere $\sphereY{p}{t}$, with $t \in [r,2r]$, that at least $2n/c$ endpoints of $S$ are inside (resp. outside) it, as $\ballY{p}{2r}$ can be covered by $c'$ balls of radius $r$, and each one of them contains at most $2n/c$ endpoints of $S$ -- namely, $\ballY{p}{2r}$ contains at most $(c'/c)2n$ endpoints of $S$.

    Pick a random radius $x \in [r,2r]$, and consider the sphere $\sphereC = \sphereY{p}{x}$. By \lemref{silly}, the expected number of intersections of $\sphereC$ with $S$ is
    \begin{math}
        \leq \frac{2cr}{r} \leq 2c.
    \end{math}
    By Markov's inequality, the probability that a random $x$ would have more than $4c$ intersections is at most $1/2$. If this happens, the algorithm repeat the sampling till success.. Otherwise, the algorithm is done as $\sphereC$ intersects at most $4c$ segments of $S$ and provides the desired separation, as at least $(2n/c - 4c)/2 \geq \tfrac{n}{2c}$ segments are on each side of it, for $n \geq 8c^2$.

    As for an efficient algorithm, repeat the above construction with $c=(c')^2 + 1$, where the ball used is a $2$-approximation to the smallest ball containing $2n/c$ points of $V(S)$, computed in linear time \cite{hr-nplta-15}. It is easy to argue that this ball has the desired properties.
\end{proof}

\BibLatexMode{\printbibliography}

\end{document}